\documentclass[11pt]{article}

\usepackage{algorithm} 
\usepackage[noend]{algpseudocode}
\usepackage[margin=1in]{geometry}
\usepackage{setspace}
\usepackage{lineno}

\algrenewcommand\algorithmicrequire{\textbf{Input:}}
\algrenewcommand\algorithmicensure{\textbf{Output:}}

\usepackage[T1]{fontenc}
\usepackage[utf8]{inputenc}
\usepackage{lmodern}
\usepackage{microtype}

\usepackage{amsmath,amssymb,amsthm,mathtools}
\usepackage{bm}

\usepackage[backend=biber,style=numeric,maxbibnames=100,maxalphanames=4,sortcites=true]{biblatex}

\bibliography{general,stacks,queues,track}

\usepackage[dvipsnames,svgnames,x11names,table]{xcolor}
\usepackage[colorlinks]{hyperref}
\usepackage{nicefrac}
\usepackage{graphicx}
\usepackage{algorithm}
\usepackage{authblk}
\usepackage{esvect}
\usepackage[normalem]{ulem}
\usepackage{pgfplots}
\usepackage{multicol}
\usepackage[capitalize,nameinlink,noabbrev]{cleveref}
\usepackage{subcaption}
\usepackage[font=small]{caption}
\usepackage[toc,page]{appendix}
\usepackage{makecell}
\usepackage{array}
\usepackage[inline]{enumitem}
\usepackage{booktabs}

\newcolumntype{C}[1]{>{\centering\arraybackslash}p{#1}}

\pgfplotsset{compat=1.6}

\hypersetup{
	breaklinks = true,
	colorlinks = true,
	citecolor = teal,
	urlcolor = purple,
}

\newtheorem{theorem}{Theorem}
\newtheorem{lemma}{Lemma}

\Crefname{theorem}{Theorem}{Theorems}
\Crefname{theorem2}{Theorem}{Theorems}
\Crefname{lemma}{Lemma}{Lemmas}
\Crefname{lemma2}{Lemma}{Lemmas}
\Crefname{figure}{Fig.}{Figs.}
\Crefname{section}{Section}{Sections}
\Crefname{observation}{Observation}{Observations}
\Crefname{property}{Property}{Properties}
\Crefname{lemma}{Lemma}{Lemmas}
\Crefname{claim}{Claim}{Claims}
\Crefname{claimx}{Claim}{Claims}
\Crefname{figure}{Fig.}{Figs.}
\Crefname{subfigure}{Fig.}{Figs.} 
\Crefname{minipage}{Fig.}{Figs.}
\Crefname{enumi}{Property}{Properties}
\Crefname{reductionrule}{Rule}{Rules}
\Crefname{algorithm}{Algorithm}{Algorithms}

\usepackage{xcolor}
\usepackage{soul}
\definecolor{realblue}{rgb}{0,0,1}
\definecolor{defblue}{rgb}{0.274,0.392,0.666}
\definecolor{darkerblue}{rgb}{0.094,0.455,0.804}
\definecolor{darkblue}{rgb}{0.063,0.306,0.545}
\definecolor{linkblue}{rgb}{0.098,0.098,0.4392}
\definecolor{red}{rgb}{0.627,0.117,0.156}
\definecolor{green}{RGB}{51, 153, 102}
\definecolor{orange}{rgb}{0.903,0.739,0.382}
\definecolor{realred}{rgb}{1,0,0}
\definecolor{lipicsblue}{rgb}{0.08235294118,0.3098039216,0.537254902}

\usepackage{xspace} 
\DeclareTextFontCommand{\emph}{\color{defblue}\em}
\DeclareTextFontCommand{\bl}{\color{lipicsblue}}
\hypersetup{colorlinks=true,
	linkcolor=lipicsblue,
	anchorcolor=lipicsblue,
	citecolor=lipicsblue,
	filecolor=lipicsblue,
	menucolor=lipicsblue,
	urlcolor=lipicsblue,
	bookmarksopen=true,
	bookmarksopenlevel=2,
	bookmarksnumbered=true,
	plainpages=false,
}

\usepackage{tocloft}
\title{Product Structure Meets Track Layouts}

\author[1]{Michael~A.~Bekos}
\author[2]{Giordano~Da~Lozzo}
\author[3]{Petr~Hliněný}
\author[4]{Michael~Kaufmann}

\affil[1]{University of Ioannina, Ioannina, Greece\\\texttt{bekos@uoi.gr}}
\affil[2]{Roma Tre University, Rome, Italy\\\texttt{giordano.dalozzo@uniroma3.it}}
\affil[3]{Masaryk University, Brno, Czech Republic\\\texttt{hlineny@fi.muni.cz}}
\affil[4]{University of T{\"u}bingen, T{\"u}bingen, Germany\\\texttt{michael.kaufmann@uni-tuebingen.de}}

\date{}

\usepackage{amsmath,amssymb,amsthm}
\usepackage{xspace}
\usepackage{amssymb}

\DeclareMathOperator{\tn}{tn}

\begin{document}

\maketitle

\begin{abstract}
    A {\em track layout} of a graph is a partition of its vertices into linearly ordered independent sets, called {\em tracks}, such that no two edges between the same pair of tracks cross.
    Given a graph, the goal in this context is to determine its {\em track number}, that is, the minimum number of tracks required for the graph to admit a track layout.
    
    In this work, we present upper bounds on the track number of graphs admitting a product structure. Our main contribution is an algorithm that computes a track layout with at most $(2h+1) \cdot r \cdot \tn(H)$ tracks for every subgraph of the strong product $P^h \boxtimes K_r \boxtimes H$, where $P^h$ is the $h$-th power of a path $P$, $K_r$ is the complete graph on $r$ vertices, and $H$ is a graph with track number $\tn(H)$.
    Combined with existing product-structure results from the literature, this algorithm yields upper bounds on the track number of several graph classes.
    For planar graphs, the obtained bound matches the current best-known upper bound of $225$.
    For $1$-planar and optimal $2$-planar graphs, our algorithm yields track layouts with at most $375$ tracks, while for genus-$k$, $k$-planar, $k$-framed, $k$-map, and $k$-string graphs it provides track layouts with a number of tracks that depends solely on $k$, thus establishing new upper bounds on the track number for these graph classes.  
    The algorithm runs in linear time for planar graphs and, more generally, in $O(n + h \cdot r \cdot t +  f_t(H))$ time whenever a corresponding product-structure decomposition of the input $n$-vertex graph is provided as part of the input, where $t=\tn(H)$ and $f_t(H)$ is the time needed to compute a $t$-track layout of $H$. 
    Furthermore, our algorithm only uses elementary linked-list data structures.
\end{abstract}

\section{Introduction}
\label{sec:introduction}

Track layouts form a central notion in graph drawing and topological graph theory. In a \emph{track layout} of a graph~\cite{DBLP:journals/dmtcs/DujmovicPW04}, vertices are partitioned into linearly ordered independent sets, called \emph{tracks}, such that no two edges between the same pair of tracks \emph{cross}, i.e., their endpoints appear in opposite orders along the two tracks. The minimum number of tracks required for a graph $G$ to admit a track layout is its \emph{track~number} and is denoted by $\tn(G)$. 
As a graph parameter, track number has been studied extensively because of its close relationships with several other important graph parameters. First, graphs with bounded track number admit three-dimensional grid drawings of linear volume, and conversely such drawings imply bounded track number~\cite{DBLP:journals/siamcomp/DujmovicMW05,DBLP:journals/dmtcs/DujmovicPW04}. Second, graphs with bounded track number have bounded queue number (see, e.g.,~\cite{DBLP:journals/siamdm/HeathLR92,DBLP:journals/siamcomp/HeathR92}) and vice versa up to suitable refinements~\cite{DBLP:journals/dmtcs/DujmovicPW04}. Third, graphs with bounded track number have bounded acyclic chromatic number, that is, they admit vertex colorings with a bounded number of colors that avoid bichromatic cycles~\cite{DBLP:journals/dmtcs/DujmovicPW04}. 
These connections make track number a natural meeting point of Graph Drawing, Graph Theory, and Linear Layouts.

Since determining the track number of a graph is NP-hard~\cite{DBLP:conf/gd/BannisterDDEW16}, a large body of research has focused on establishing bounds on the track number of specific graph classes, that is, on the maximum track number over all members of the class; see \cref{tab:track-bounds} for a summary. In this context, planar graphs have played a central role, since the question of whether they have bounded track number remained open for several years; see, e.g.,~\cite{DBLP:journals/dmtcs/DujmovicPW04}. The question was resolved affirmatively by Dujmovic et al.~\cite{DujmovicJoretMicekMorinUeckerdtWood2020}, who proved that planar graphs have bounded queue number. 
A subsequent slight improvement was given by Bekos, Gronemann, and Raftopoulou~\cite{DBLP:journals/algorithmica/BekosGR23}. The current best-known upper bound is $225$, due to Pupyrev~\cite{DBLP:journals/jgaa/Pupyrev20}.

\begin{table}[t]
\centering
\small
\begin{tabular}{lcccc}
\toprule
\textbf{Graph Class} & \textbf{Lower Bound} & \textbf{Ref.} & \textbf{Upper Bound} & \textbf{Ref.} \\
\midrule
Trees & 3 & \cite{DBLP:journals/jgaa/FelsnerLW03} & 3 & \cite{DBLP:journals/jgaa/FelsnerLW03} \\
Level planar graphs & 3 & \cite{DBLP:journals/algorithmica/BannisterDDEW19} & 3 & \cite{DBLP:journals/algorithmica/BannisterDDEW19}\\
Outerplanar graphs & 5 & \cite{DBLP:journals/dmtcs/DujmovicPW04} & 5 & \cite{DBLP:journals/dmtcs/DujmovicPW04} \\
Series--parallel graphs & 7 & \cite{DBLP:journals/siamcomp/DujmovicMW05} & 15 & \cite{DBLP:journals/comgeo/GiacomoLM05} \\
Planar 3-trees & 8 & \cite{DBLP:journals/jgaa/Pupyrev20} & 25 & \cite{DBLP:journals/jgaa/Pupyrev20} \\
Planar graphs & 8 & \cite{DBLP:journals/jgaa/Pupyrev20} & 225 & \cite{DBLP:journals/jgaa/Pupyrev20} \\
Treewidth-$k$ graphs & $\frac{1}{2}(k + 1)(k + 2) + 1$ & \cite{DBLP:journals/siamcomp/DujmovicMW05} & $(k + 1)(2^{k+1}-2)^k$ & \cite{DBLP:journals/combinatorics/Wiechert17}\\
\bottomrule
\end{tabular}
\caption{Known lower and upper bounds on the track number for selected graph classes.}
\label{tab:track-bounds}
\end{table}

A key ingredient underlying these developments is a structural result that expresses planar graphs in terms of simpler graph classes. Namely, every planar graph is a subgraph of a graph isomorphic to the strong product $P \boxtimes K_3 \boxtimes H$ of a path, a clique on three vertices, and a planar 3-tree~\cite{DBLP:conf/focs/DujmovicEGJMM20}. While this point of view has proved particularly fruitful, and has since been extended and refined for numerous graph classes beyond planarity, see for example,~\cite{DBLP:journals/corr/abs-2001-08860,DBLP:journals/combinatorics/BekosLHK24}, its 
connection with track number has not previously been formulated explicitly in product-structure terms\footnote{At this point, we deem it important to mention that Pupyrev’s proof that planar graphs have track number at most $225$ can, in hindsight, be interpreted as relying on a product-structure viewpoint~\cite{DBLP:journals/jgaa/Pupyrev20}; however, this connection is not made explicit, and the arguments in the proof are technically involved and not readily reusable in that form.}. As a result, the current state of the art has largely remained restricted to planar graphs and has not yet been extended beyond this setting, even though the underlying structural theory is mature enough to support such~extensions.

\medskip\noindent\textbf{Our contribution.} 
In this work, we close this gap in the literature. Our main contribution is an algorithm that constructs a track layout with at most $(2h+1) \cdot r \cdot t$ tracks for the strong product $P^h \boxtimes K_r \boxtimes H$ of the $h$-th power of a path, a complete graph on $r$ vertices, and a graph $H$ admitting a $t$-track layout; see \cref{sec:combinatorial}. 
 An immediate consequence is the following combinatorial result, which constitutes the main body of this work.
\begin{theorem}\label{thm:main}
	For each graph $G$ such that $G \subseteq P^h \boxtimes K_r \boxtimes H$, it holds that $\tn(G) \leq (2h+1) \cdot r \cdot \tn(H)$.
\end{theorem}
In contrast to several related results in the field, the correctness of our algorithm and proof follows from a surprisingly simple and intuitive argument (almost amounting to a proof by a picture).
Furthermore, the algorithm is straightforward to grasp, its implementation does not require the use of any complex data structure, and it runs in $O(n + h \cdot r \cdot t +  f_t(H))$ time whenever a corresponding product-structure decomposition of the input $n$-vertex graph is provided as part of the input or can be computed in linear time (as is the case, e.g., for planar graphs~\cite{DBLP:conf/swat/BoseMO22}), where $f_t(H)$ is the time needed to compute a $t$-track layout of $H$; see \cref{sec:algorithmics}. 

More importantly, our contribution provides a way to easily translate results from the product-structure literature (which is rather extensive nowadays) into efficient upper bounds on the track number; see~\cref{tbl:results} for a brief summary. It is worth noting that, for planar graphs, the obtained bound matches the current best-known upper bound of $225$ by Pupyrev~\cite{DBLP:journals/jgaa/Pupyrev20}. 
For $1$-planar and optimal $2$-planar graphs, our algorithm yields track layouts with at most $375$ tracks, while for genus-$k$, $k$-planar, $k$-framed, $k$-map, and $k$-string graphs it produces track layouts with a number of tracks that depends solely on $k$, thus establishing new best-known upper bounds on the track number of these graph classes. Another important consequence of our result is that, for every fixed $k$, the $n$-vertex graphs belonging to these classes admit three-dimensional grid drawings with $O(1)\times O(1) \times O(n)$, and hence linear volume~\cite{DBLP:journals/siamcomp/DujmovicMW05}.

The remainder of the paper is organized as follows.
\cref{se:preliminaries} introduces the necessary definitions and notation. \cref{sec:combinatorial} presents the track-layout constructions and proves the main theorem. \cref{sec:algorithmics} describes their algorithmic implementation and analyzes the running time. \cref{sec:conclusions} concludes the paper with final remarks and lists some relevant open problems that stem from our research.

\begin{table}[t]
\centering
\small
\begin{tabular}{|p{3.8cm}|p{3.8cm} c| >{\centering\arraybackslash}p{6.1cm}|}
\hline
\textbf{Graph Class} & \textbf{Product Structure}  & \textbf{Ref.} & \textbf{Track Number Upper Bound} \\
\hline
Planar graphs & \mbox{\vbox to 2.3ex{}$P \boxtimes K_3 \boxtimes H$} \mbox{$H$: Planar $3$-tree} &  \cite{DBLP:conf/focs/DujmovicEGJMM20} & $225$ \\
\hline
\mbox{1-planar graphs}
\mbox{Optimal 2-planar graphs} & \mbox{\vbox to 2.4ex{}$P^2 \boxtimes K_3 \boxtimes H$}  \mbox{$H$: Planar $3$-tree} &  \cite{DBLP:journals/combinatorics/BekosLHK24,DBLP:conf/compgeom/Bekos0R17} & $375$ \\
\hline
Optimal 3-planar graphs & \mbox{\vbox to 2.4ex{}$P^3 \boxtimes K_4 \boxtimes H$}  \mbox{$H$: Planar $3$-tree} &  \cite{DBLP:journals/combinatorics/BekosLHK24,DBLP:conf/compgeom/Bekos0R17} & $700$ \\
\hline
genus-$k$ graphs & \mbox{\vbox to 2.4ex{}$P \boxtimes K_{\max\{2k,3\}} \boxtimes H$} \mbox{$H$: $4$-tree} &  \cite{DBLP:conf/focs/DujmovicEGJMM20} & $15\cdot30^4\cdot \max\{2k,3\}$ \\
\hline
$k$-framed, $k$-map graphs & \mbox{\vbox to 2.6ex{}$P^{\lfloor \frac{k}{2} \rfloor} \boxtimes K_{\max\{3,k-2\}} \boxtimes H$} \mbox{$H$: Planar $3$-tree} &  \cite{DBLP:journals/combinatorics/BekosLHK24} & $25 \cdot \left(2\left\lfloor \frac{k}{2} \right\rfloor + 1\right)\cdot \max\{3,\, k-2\}
$ \\
\hline
$k$-planar graphs & \mbox{\vbox to 2.4ex{}$P \boxtimes K_{18k^2+48k+30} \boxtimes H$} \mbox{$H$: Treewidth $\binom{k+4}{3}$\vtop to 1.3ex{}} &  \cite{DBLP:journals/jctb/DujmovicMW23} & \mbox{$3\cdot(18k^2+48k+30) \cdot (\lambda + 1)(2^{\lambda+1}-2)^{\lambda} $} \mbox{where $\lambda = \binom{k+4}{3}$} \\
\hline
$k$-string graphs & \mbox{\vbox to 2.6ex{}$H \boxtimes P$} \mbox{$H$: Treewidth $O(k^7)$} &  \cite{DBLP:journals/jctb/DujmovicMW23} & $O(2^{k^{14}})$ \\
\hline
\end{tabular}
\caption{Product-structure decompositions and corresponding bounds on the track number derived from \cref{thm:main}. Refer to \cref{se:preliminaries} for the definitions of the considered graph families.}
\label{tbl:results}
\end{table}

\section{Preliminaries}
\label{se:preliminaries}

In this section, we give preliminaries and basic definitions that will be used in the following. 
For standard graph theoretic definitions and notation, refer e.g. to~\cite{DBLP:books/daglib/0030488}.

\subparagraph{Basic definitions.}
A graph is \emph{simple} if it contains neither loops nor multi-edges. 
For any $i\geq 1$, the \emph{$i$-th power} $G^i$ of a graph $G$ is the graph with the same vertex set as $G$, in which two distinct vertices are adjacent if and only if they are at distance at most~$i$ in~$G$. Clearly, $G \subseteq G^i$. The \emph{chromatic number} $\chi(G)$ of a graph $G$ is the smallest integer $k$ such that the vertices of $G$ can be colored with $k$ colors so that no two adjacent vertices receive the same color.

\subparagraph{Graph classes.} For completeness, we briefly define the graph classes listed in \cref{tbl:results}; see, e.g., \cite{DBLP:journals/jctb/DujmovicMW23} for further details. A \emph{planar} graph is a graph that can be embedded in the plane so that no two of its edges cross. More generally, a \emph{genus-$k$}  graph is a graph that can be embedded on an orientable surface of genus $k$ without crossings, while a graph is \emph{$k$-planar} if it admits a drawing in the plane in which each edge is crossed at most $k$ times. An $n$-vertex $k$-planar graph that has the maximum number of edges among all $n$-vertex $k$-planar graphs is called \emph{optimal}. A graph is said to be \emph{$k$-framed} if it admits a drawing in the plane in which the crossing-free edges determine a simple, biconnected graph that spans all vertices and has all faces of size at most $k$, in the interior of which lie all edges involved in crossings.
A \emph{$k$-map} graph is a graph whose vertices correspond to connected regions of the plane, where each region is adjacent to at most $k$ other regions, and edges represent adjacencies between regions. Finally, a graph is a \emph{$k$-string} graph if it is the intersection graph of a family of curves in the plane such that any two curves intersect in at most $k$ points and no three curves pass through the same point.

\subparagraph{Product structure.} 

The \emph{strong product} of two graphs $X$ and $Y$, denoted by $X \boxtimes Y$, is the graph whose vertex set is the Cartesian product of the vertex sets of $X$ and $Y$, such that there exists an edge in $X \boxtimes Y$ between the vertices $\langle x_i,y_k \rangle, \langle x_j,y_\ell \rangle$ if and only if one of the following occurs: (a)~$x_i = x_j$ and $(y_k,y_\ell)$ is an edge of $Y$, (b)~$y_k = y_\ell$ and $(x_i,x_j)$ is an edge of $X$, or (c)~$(x_i,x_j)$ is an edge of $X$ and $(y_k,y_\ell)$ is an edge of $Y$; see \cref{fig:strong}.

\begin{figure}[htb]
	\centering
	\includegraphics[page=2]{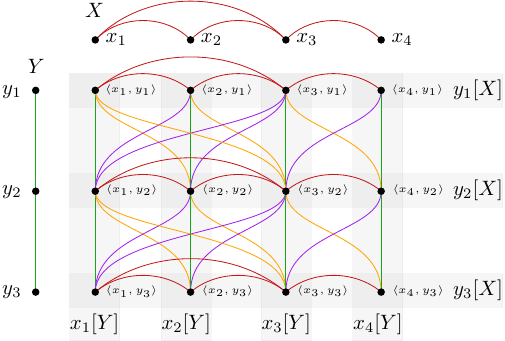}
	\caption{The strong product $X \boxtimes Y$ of a planar graph $X$ (red) and a path~$Y$ (green).}
	\label{fig:strong}
\end{figure}

\noindent Equivalently, and more conveniently for our proofs, the graph $X \boxtimes Y$ can be described as follows. 

\begin{enumerate}[label=(\alph*),ref=(\alph*)]
	\item\label{s:hor} For each vertex $y_k$ of $Y$, replace $y_k$ by a copy of $X$, denoted by 
	$y_k[X]$. Then $y_k[X]$ induces a subgraph isomorphic to $X$, whose vertices we denote by $\langle x_i,y_k \rangle$ with $x_i$ being a vertex of $X$.     
	\item For every edge $(y_k,y_\ell)$ of $Y$ with $k<\ell$, we add edges between the copies $y_k[X]$ and $y_\ell[X]$ as follows:
	\begin{enumerate*}[label=(b.\arabic*),ref=(b.\arabic*)]
		\item\label{s:vert}each vertex $\langle x_i,y_k \rangle$ is connected to $\langle x_i,y_\ell \rangle$, and whenever $x_i$ and $x_j$ with $i<j$ are adjacent in $X$, we also connect
		\item\label{s:forward}$\langle x_i,y_k \rangle$ with $\langle x_j,y_\ell \rangle$, and 
		\item\label{s:backward}$\langle x_j,y_k \rangle$ with $\langle x_i,y_\ell \rangle$.
	\end{enumerate*}
\end{enumerate}

\noindent Or equivalently:

\begin{enumerate}[label=(\alph*),ref=(\alph*), start=3]
	\item\label{t:vert} For each vertex $x_i$ of $X$, replace $x_i$ by a copy of $Y$, denoted by 
	$x_i[Y]$. Then $x_i[Y]$ induces a subgraph isomorphic to $Y$, whose vertices we denote by $\langle x_i,y_k \rangle$ with $y_k$ being a vertex of $Y$.     
	\item For every edge $(x_i,x_j)$ of $X$ with $i<j$, we add edges between the copies $x_i[Y]$ and $x_j[Y]$ as follows:
	\begin{enumerate*}[label=(d.\arabic*),ref=(d.\arabic*)]
		\item\label{t:hor}each vertex $\langle x_i,y_k \rangle$ is connected to $\langle x_j,y_k \rangle$, and whenever $y_k$ and $y_\ell$ with $k<\ell$ are adjacent in $Y$, we also connect
		\item\label{t:forward}$\langle x_i,y_k \rangle$ with $\langle x_j,y_\ell \rangle$, and 
		\item\label{t:backward}$\langle x_j,y_k \rangle$ with $\langle x_i,y_\ell \rangle$.
	\end{enumerate*}
\end{enumerate}

The above definitions provide a partition of the edges of $X \boxtimes Y$ into four sets: \ref{s:hor}, \ref{s:vert}, \ref{s:forward} and \ref{s:backward}, or equivalently \ref{t:hor}, \ref{t:vert}, \ref{t:forward} and \ref{t:backward}, respectively. We refer to the edges of these four sets as \emph{horizontal edges} (red in \cref{fig:strong}), \emph{vertical edges} (green in \cref{fig:strong}), \emph{forward edges} (orange in \cref{fig:strong}), and \emph{backward edges} (purple in \cref{fig:strong}), respectively.

\subparagraph{Track layouts.} 

Let $\{V_i:\; 1\le i \le t\}$ be a partition of the vertex set of a graph $G$ such that for every edge $(u,v)$ of $G$, if $u \in V_i$ and $v \in V_j$, then $i \neq j$. Suppose that $\prec_i$ is a total order of $V_i$. Then, the ordered set $(V_i, \prec_i)$ is called a \emph{track} and the ordered partition $\{(V_i, \prec_i)\}^t_{i=1}$ is called a \emph{t-track assignment} of $G$. An \emph{X-crossing} in a track assignment consists of two edges $(u,v)$ and $(u',v')$  such that $u$ and $u'$ are on the same track $V_i$, $v$ and $v'$ are on a different track $V_j$ with $u \prec_i u'$ and $v' \prec_j v$. A \emph{$t$-track layout} is a $t$-track assignment with no X-crossings. 
In other words, a track layout of a graph is a partition of its vertices into a number of tracks, such that the vertices in each track form an independent set and the edges between each pair of tracks form a non-crossing set. The \emph{track number} of a graph $G$, denoted by $\tn(G)$, is the minimum $t$ such that $G$ has a $t$-track layout. 

\section{Track Number of Graphs Having a Product Structure}
\label{sec:combinatorial}

In order to show our main result, we start by providing upper bounds on the track number of subgraphs of the strong product of two graphs, one of which has a special structure. In the case of such a graph being a path, as in the next lemma, our upper bound exploits the fact that the chromatic number of the second power of a path is $3$, i.e., $\chi(P^2) = 3$.

\begin{lemma}\label{lem:path}
	For each graph $G$ such that $G \subseteq P \boxtimes H$, where $P$ is a path, it holds that $\tn(G)\leq 3\tn(H)$. 
\end{lemma}
\begin{proof}
	Let $\mathcal{L}$ be a track layout of $H$ with $t=\tn(H)$ tracks and let $P=(x_1,\dots,x_p)$.
	We define a $3t$-track layout $\Psi$ for $H \boxtimes P$ as follows; refer to \cref{fig:kn}. 
	Let $L_0,\ldots,L_{3t-1}$ be the tracks of the resulting layout, which initially are empty.
	For each vertex $x_i$ of $P$ and $q = i \!\mod 3$, we create a copy of $\mathcal{L}$, which we denote by $\mathcal{L}_i$, at the following $t$ tracks: $L_q,L_{q+3},\ldots,L_{q+3(t-1)}$. Obviously, this copy defines a track layout of $x_i[H]$. In order to define a total order of the vertices assigned to each track, for each pair of copies $\mathcal{L}_i$ and $\mathcal{L}_j$ assigned to the same $t$ tracks (that is, $i \equiv j \mod 3$), we additionally require that, in each of these $t$ tracks, the vertices of $\mathcal{L}_i$ 
	precede all vertices of $\mathcal{L}_j$ if and only if $i < j$. Hence, the obtained track assignment $\Psi$ of $P \boxtimes H$ has the following properties:  
	\begin{enumerate}[label=P.\arabic*,ref=P.\arabic*]
		\item\label{p:different-tracks}For each $i\in\{1,\ldots,p\}$ and $j\in\{1,\ldots,p\}$ with $i \not\equiv j \mod 3$, the vertices of $x_i[H]$ are assigned to different tracks from the vertices of $x_{j}[H]$, and 
		\item\label{p:common-tracks}For each $i\in\{1,\ldots,p\}$ and $j\in\{1,\ldots,p\}$ with $i \neq j$ and $i \equiv j \mod 3$, the vertices of $x_i[H]$ share the same tracks as the vertices of $x_j[H]$. Furthermore, if $i<j$, in each of these tracks, the vertices of $x_i[H]$ precede the vertices of $x_j[H]$.
	\end{enumerate}
	Observe that \cref{p:different-tracks} implies that for each $i \in \{2,\ldots,p-1\}$, the vertices of $x_i[H]$ are assigned to different tracks from the vertices of each of $x_{i-1}[H]$ and $x_{i+1}[H]$. %
	To show that $\Psi$ is a track layout of ${P \boxtimes H}$, it remains to argue that it contains no $X$-crossing.
	First, we show that no pair of edges belonging to the same set (horizontal, vertical, forward, and backward) form an $X$-crossing.
	
	\begin{figure}
		\centering
		\includegraphics[page=3]{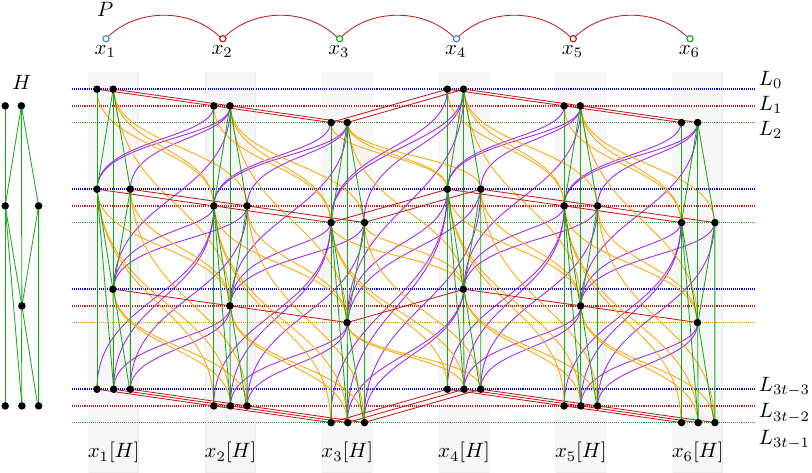}
		\caption{Illustration for the proof of \cref{lem:path}.}
		\label{fig:kn}
	\end{figure}
	
	\begin{itemize}
		\item \textit{No two vertical edges are involved in an $X$-crossing} (refer to the green edges of \cref{fig:kn}): By \cref{p:different-tracks}, an $X$-crossing among a pair of vertical edges $e$ and $e'$ could exist only if $e$ connects two vertices in $x_i[H]$, $e'$ connects two vertices in $x_j[H]$, and $i \equiv j \mod 3$ holds, as otherwise $e$ and $e'$ do not have their endpoints on the same two tracks. If $i=j$, such an $X$-crossing is not possible, since the vertices belonging to the same copy of $H$ are ordered within each track as in $\mathcal{L}$. Otherwise, if we assume w.l.o.g.\ that $i < j$, then all vertices of $x_i[H]$ precede all vertices of $x_j[H]$ in each relevant track of $\Psi$ (by \cref{p:common-tracks}). Hence, again $e$ and $e'$ cannot form an $X$-crossing.
		
		\item \textit{No two horizontal edges are involved in an $X$-crossing} (refer to the red edges of \cref{fig:kn}): By \cref{p:different-tracks}, an $X$-crossing among a pair of horizontal edges $e$ and $e'$ could exist only if $e$ connects a vertex of $x_i[H]$ to a vertex in $x_{i+1}[H]$, $e'$ connects a vertex in $x_j[H]$ to a vertex in $x_{j+1}[H]$, and $i \equiv j \mod 3$, as otherwise $e$ and $e'$ do not have their endpoints on the same two tracks. If $i=j$, such an $X$-crossing is not possible, because the horizontal edges connecting vertices in $x_i[H]$ and $x_{i+1}[H]$ (or equivalently in $x_j[H]$ and $x_{j+1}[H]$) form a matching and their endpoints have the same relative order in $\mathcal{L}_i$ and $\mathcal{L}_{i+1}$. Otherwise, assuming w.l.o.g.\ that $i < j$, we know that the endpoint of $e$ in $x_i[H]$ precedes the endpoint of $e'$ in $x_j[H]$, while the endpoint of $e$ in $x_{i+1}[H]$ precedes the endpoint of $e'$ in $x_{j+1}[H]$ (by \cref{p:common-tracks}). Therefore, $e$ and $e'$ cannot form an $X$-crossing. 
		
		\item \textit{No two forward edges are involved in an $X$-crossing} (refer to the orange edges of \cref{fig:kn}): By \cref{p:different-tracks}, an $X$-crossing among a pair of forward edges $e$ and $e'$ could exist only if $e$ connects a vertex of $x_i[H]$ to a vertex in $x_{i+1}[H]$, $e'$ connects a vertex in $x_j[H]$ to a vertex in $x_{j+1}[H]$, and $j \equiv i \mod 3$, as otherwise $e$ and $e'$ do not have their endpoints on the same two tracks. We first argue that if $i=j$, then $e$ and $e'$ cannot form an $X$-crossing. To see this, assume for a contradiction that $e$ and $e'$ are two such forward edges that form an $X$-crossing. Since each of them connect vertices in $x_i[H]$ and $x_{i+1}[H]$ (or equivalently in $x_j[H]$ and $x_{j+1}[H]$), we can assume w.l.o.g. that $e=(\langle x_i, y_k\rangle, \langle x_{i+1},y_\ell\rangle)$ and $e'=(\langle x_i, y_{k'}\rangle, \langle x_{i+1},y_{\ell'}\rangle)$.
		W.l.o.g., we assume that $\langle x_i, y_k\rangle \prec \langle x_i, y_{k'}\rangle$, which implies that, in the track layout $\mathcal{L}$ of $H$, it holds that $y_k$ precedes $y_{k'}$. Since $e$ and $e'$ form an $X$-crossing, it follows that $\langle x_{i+1},y_{\ell'}\rangle$ precedes $\langle x_{i+1},y_{\ell}\rangle$, which implies that, in the track layout $\mathcal{L}$ of $H$, it holds that $y_{\ell'}$ precedes $y_{\ell}$. However, in this scenario the edges $(y_k,y_\ell)$ and $(y_{k'},y_{\ell'})$ of $H$ form an $X$-crossing in $\mathcal{L}$; a contradiction. Hence, $i \neq j$ holds. In this case, however, $e$ and $e'$ cannot form an $X$-crossing either, since we know that the endpoint of $e$ in $x_i[H]$ precedes the endpoint of $e'$ in $x_j[H]$ if and only if the endpoint of $e$ in $x_{i+1}[H]$ precedes the endpoint of $e'$ in $x_{j+1}[H]$ (by \cref{p:common-tracks}).
		\item \textit{No two backward edges are involved in an $X$-crossing} (refer to the purple edges of \cref{fig:kn}): The argument here is symmetric to the one of forward edges.
	\end{itemize}
	
	\noindent Second, we show that no pair of edges belonging to distinct sets forms an $X$-crossing in $\Psi$.
	
	\begin{itemize}
		\item \textit{No horizontal edge is involved in an $X$-crossing with any vertical, forward, or backward edge}: Both endpoints of a horizontal edge belong to tracks $L_{3i}$, $L_{3i+1}$, and $L_{3i+2}$, for some $i \in \{0,\dots,t-1\}$. Since each vertical, forward, or backward edge has at most one endpoint in such tracks, an $X$-crossing cannot be formed in this case.
		
		\item \textit{No vertical edge is involved in an $X$-crossing with a forward or a backward edge}:
		To see this, assume for a contradiction that $e$ and $e'$ are two such independent edges that form an $X$-crossing such that $e$ is vertical and $e'$ is forward; the case in which $e'$ is backward is symmetric. Since $e$ is vertical, both its endpoints belong to some graph $x_i[H]$. So, we can assume w.l.o.g. that $e=(\langle x_i, y_k\rangle, \langle x_i,y_\ell\rangle)$. 
		Since $e'$ is forward and since $e'$ forms an $X$-crossing with $e$, we can assume, by the total order of vertices on the tracks of $\Psi$, that $e'=(\langle x_i, y_{k'}\rangle, \langle x_{i+1},y_{\ell'}\rangle)$ or $e'=(\langle x_{i-1}, y_{k'}\rangle, \langle x_i,y_{\ell'}\rangle)$. W.l.o.g.\ assume the former case; the latter one is symmetric. Since $e$ and $e'$ cross, the vertices $\langle x_i, y_k\rangle$ and $\langle x_i, y_{k'}\rangle$ must be on the same track.  
		Since $e$ is vertical, vertex $\langle x_i,y_\ell\rangle$ belongs to $x_i[H]$. On the other hand, since $e'$ is forward, vertex $\langle x_{i+1},y_{\ell'}\rangle$ belongs to $x_{i+1}[H]$. By \cref{p:different-tracks}, it follows that $\langle x_i,y_\ell\rangle$ and $\langle x_{i+1},y_{\ell'}\rangle$ are on different tracks. Hence, $e$ and $e'$ cannot form an $X$-crossing.
		
		\item \textit{A forward edge cannot form an $X$-crossing with a backward edge}: By \cref{p:different-tracks}, an $X$-crossing among such a pair of edges $e$ and $e'$ with $e$ being forward and $e'$ being backward exists only if $e$ connects a vertex of $x_i[H]$ to a vertex in $x_{i+1}[H]$, $e'$ connects a vertex in $x_j[H]$ to a vertex in $x_{j+1}[H]$, and $j \equiv i \mod 3$, as otherwise $e$ and $e'$ do not have their endpoints on the same two tracks. Assume first that $i=j$. Since $e$ is forward it connects $\langle x_i,y_k \rangle$ with $\langle x_{i+1},y_\ell \rangle$ with $k < \ell$. Symmetrically, since $e'$ is backward it connects $\langle x_{i+1},y_{k'} \rangle$ with $\langle x_i,y_{\ell'} \rangle$  with $k' < \ell'$. For $e$ and $e'$ to form an $X$-crossing, $k=\ell'$ and $\ell = k'$ must hold by our track assignment. However, in this case we obtain $k < \ell = k' < \ell' = k$, which is a contradiction.  Hence, $i \neq j$ holds. In this case, however, $e$ and $e'$ cannot form an $X$-crossing either, since we know that the endpoint of $e$ in $x_i[H]$ precedes the endpoint of $e'$ in $x_j[H]$ if and only if the endpoint of $e$ in $x_{i+1}[H]$ precedes the endpoint of $e'$ in $x_{j+1}[H]$ (by \cref{p:common-tracks}).
	\end{itemize}
	It follows that $\Psi$ is a $3t$-track layout of $P \boxtimes H$, completing the proof.
\end{proof}

A straightforward generalization of \Cref{lem:path} is the following \Cref{lem:power} in which we similarly exploit the fact that the chromatic number of the $2h$-th power of a path is $2h+1$, i.e., $\chi(P^{2h}) = 2h+1$.

\begin{lemma}\label{lem:power}
	For each graph $G$ such that $G \subseteq P^h \boxtimes H$, it holds that $\tn(G)\leq (2h+1)\cdot\tn(H)$.
\end{lemma}
\begin{proof}
    As in the proof of \cref{lem:path}, let $\mathcal{L}$ be a track layout of $H$ with $t=\tn(H)$ tracks and let $(x_1,\dots,x_p)$ denote the path $P$.
	We define a $(2h+1)t\,$-track layout $\Psi$ for $P^h \boxtimes H$ as follows; refer to \cref{fig:kn} with $h=1$. 
	Let $L_0,\ldots,L_{(2h+1)t-1}$ be the tracks of the resulting layout, which initially are empty.
	For each vertex $x_i$ of $P$ with 
    $q = i \!\mod (2h+1)$, 
    we create a copy of $\mathcal{L}$, which we denote by $\mathcal{L}_i$, at the following $t$ tracks: $L_q,L_{q+(2h+1)},\ldots,L_{q+(2h+1)(t-1)}$. 
    Obviously, this copy defines a track layout of $x_i[H]$. In order to define a total order of the vertices assigned to each track, for each pair of copies $\mathcal{L}_i$ and $\mathcal{L}_j$ assigned to the same $t$ tracks (that is, $i \equiv j \mod (2h+1)$), we additionally require that, in each of these $t$ tracks, the vertices of $\mathcal{L}_i$ 
	precede all vertices of $\mathcal{L}_j$ if and only if $i < j$. Hence, the obtained track assignment $\Psi$ of $P^h \boxtimes H$ has the following properties:  
	\begin{enumerate}[label=P.\arabic*,ref=P.\arabic*,start=3]
		\item\label{mp:different-tracks}For each $i\in\{1,\ldots,p\}$ and $j\in\{1,\ldots,p\}$ with $i \not\equiv j \mod (2h+1)$, the vertices of $x_i[H]$ are assigned to different tracks from the vertices of $x_{j}[H]$, and 
		\item\label{mp:common-tracks}For each $i\in\{1,\ldots,p\}$ and $j\in\{1,\ldots,p\}$ with $i \neq j$ and $i \equiv j \mod (2h+1)$, the vertices of $x_i[H]$ share the same tracks as the vertices of $x_j[H]$. Furthermore, if $i<j$, in each of these tracks, the vertices of $x_i[H]$ precede the vertices of $x_j[H]$.
	\end{enumerate}
	Observe that \cref{mp:different-tracks} implies that for each $i \in \{h+1,\ldots,p-h\}$, the vertices of $x_i[H]$ are assigned to different tracks from the vertices of each of $x_{i-h}[H],\ldots,x_{i-1}[H]$ and $x_{i+1}[H],\ldots,x_{i+h}[H]$.	
    With this observation, the proof that $\Psi$ is a track layout of ${P^h \boxtimes H}$, i.e., that the designed track assignment and the vertex orders do not produce any $X$-crossing, exploits \cref{mp:different-tracks,mp:common-tracks} in the same way as \cref{p:different-tracks,p:common-tracks}, respectively, are used in the proof of \cref{lem:path}.
\end{proof}

\begin{figure}[t]
    \centering
    \includegraphics[page=4]{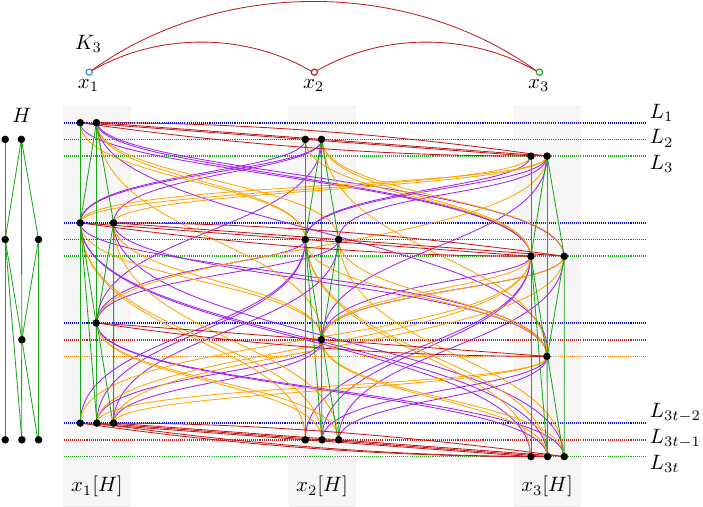}
    \caption{Illustration for the proof of \cref{lem:complete}.}
    \label{fig:complete}
\end{figure}

\begin{lemma}\label{lem:complete}
	For each graph $G$ such that $G \subseteq K_r \boxtimes H$, it holds that $\tn(G)\leq r \cdot\tn(H)$.    
\end{lemma}
\begin{proof}
    Let $h=r-1$ and $P=(x_1,\ldots,x_r)$ be a path on $r$ vertices. Then $P^h\simeq K_r$ and so, up to isomorphism, $G \subseteq P^h \boxtimes H$. We consider the $(2h+1)t\,$-track layout $\Psi$ constructed for the graph $P^h \boxtimes H$ in the proof of \Cref{lem:power}. Since in this case $P$ has only $r$ vertices, the construction of \Cref{lem:power} leaves empty all $h$ subcollections of tracks $L_q,L_{q+(2h+1)},\ldots,L_{q+(2h+1)(t-1)}$ where $q=0$ or $r=h+1<q<2h+1=r+h$ (see \cref{fig:complete}). In other words, $\Psi$ is now a $(2h+1-h)t\,$-track layout where $2h+1-h=r$, and hence $\tn(G)\leq r \cdot\tn(H)$, completing the proof.   
\end{proof}
\Cref{lem:complete,lem:power} together imply \cref{thm:main}, which constitutes the main result of this work. 
All bounds presented in \cref{tbl:results} are obtained as a corollary of \cref{thm:main} using the cited product structure theorems from the literature.

\section{Implementation Details}
\label{sec:algorithmics}
%
\begin{algorithm}
	[tb!]
	\caption{Construction of the track assignment $\Psi$ for \cref{lem:path,lem:power}.}
	\label{alg:path-track-layout}
	\begin{algorithmic}
		[1] \Require A graph $G\subseteq P^{h}\boxtimes H$, where
		$P=(x_{1},\ldots,x_{p})$, a $t$-track layout $\mathcal{L}=\{(T_{a},\prec_{a}
		):0\leq a<t\}$ of $H$, and the integer $h$ 
        \Ensure A $(2h+1)t$-track assignment $\Psi=\{(L_{\ell},\prec
		_{\ell}):0\leq\ell<(2h+1)t\}$ of $G$
        \Statex
        \For{$\ell=0,\ldots,(2h+1)t-1$}\label{lst:line:1} \State \label{lst:line:2}Initialize $L_{\ell}$ as an empty ordered list \EndFor 
        \ForAll{$y\in V(H)$} \label{lst:line:3} \State Initialize
		$C_{y}$ as an empty list \EndFor 
        \For{$i=1,\ldots,p$} \State Initialize $\mathcal{A}_{i}$ as an empty list \EndFor 
        \ForAll{$v\in V(G)$} \State Append $v$ to
		$C_{v.y}$ \EndFor 
        \For{$a=0,\ldots,t-1$} \ForAll{$y\in T_{a}$, taken according to $\prec_{a}$}
		\ForAll{$v\in C_{y}$} \State Let $v.x=x_{i}$ \If{$B_{i,a}$ has not yet been created}
		\State Initialize $B_{i,a}$ as an empty ordered list \State Append $a$ to $\mathcal{A}
		_{i}$ \EndIf \State Append $v$ to $B_{i,a}$ \EndFor \EndFor \EndFor \For{$i=1,\ldots,p$}
		\ForAll{$a\in\mathcal{A}_{i}$} \State $\ell\gets (2h+1)a+(i\bmod (2h+1))$ \State Append
		the ordered list $B_{i,a}$ to the end of $L_{\ell}$ \label{lst:line:4}\EndFor \EndFor \For{$\ell=0,\ldots,(2h+1)t-1$} \label{lst:line:5}
		\State Let $\prec_{\ell}$ be the order of the vertices in $L_{\ell}$\label{lst:line:6} \EndFor
		\State \Return
		\[
			\Psi=\{(L_{\ell},\prec_{\ell}):0\leq\ell<(2h+1)t\}
		\]
	\end{algorithmic}
\end{algorithm}
We now present an explicit algorithmic description of the construction of the track assignment $\Psi$ from \cref{lem:power}. When $h=1$, this construction yields the track assignment of \cref{lem:path}, whereas when $h=r-1$ and $p=r$, it specializes the track assignment to the one of \cref{lem:complete}. 
The same construction is also presented by \cref{alg:path-track-layout}.

Let $G\subseteq P^h\boxtimes H$ be an $n$-vertex graph, let $P=(x_1,\ldots,x_p)$, and let $\mathcal L=\{(T_a,\prec_a):0\leq a<t\}$  be a $t$-track layout of $H$. 
Since $G\subseteq P^h\boxtimes H$, every vertex $v\in V(G)$ corresponds to a unique ordered pair $\langle x_i, y_j\rangle$, where $x_i\in V(P)$ and $y_j\in V(H)$. We assume that $v$ stores these two coordinates in the fields $v.x$ and $v.y$, respectively. Thus, $v.x=x_i$ and $v.y=y_j$. Number the tracks of $\mathcal L$ from $0$ to $t-1$, and let $t_{\mathcal L}\colon V(H)\to\{0,\ldots,t-1\}$ be the function that maps each vertex $y\in V(H)$ to the index of the track containing $y$ in $\mathcal L$, that is, $T_a = \{y \in V(H): t_{\mathcal L}(y) = a\}$.   
Also, for every vertex $x_i$ of $P$, define $\operatorname{ind}_P(x_i):=i$. The 
function that maps each vertex $v \in V(G)$ to the index of the track containing $v$ in $\Psi$, which according to \cref{lem:power}, 
is then \[ t_\Psi(v) :=(2h+1)\,t_{\mathcal L}(v.y) +\bigl(\operatorname{ind}_P(v.x)\bmod (2h+1)\bigr). \] Consequently, for each $\ell\in\{0,\ldots,(2h+1)t-1\}$, the track $L_\ell$ of $\Psi$ contains precisely the vertices $v\in V(G)$ satisfying $L_\ell := \left\{ v\in V(G): 
t_\Psi(v) =\ell \right\}$.

It remains to define the total order of the vertices assigned to each track. For every $i\in\{1,\ldots,p\}$ and $a\in\{0,\ldots,t-1\}$, let \[ B_{i,a} = \left\{ v\in V(G): \operatorname{ind}_P(v.x) =i \text{ and } t_{\mathcal L}(v.y)=a \right\}. \] We refer to $B_{i,a}$ as a \emph{bucket}. Observe that $L_\ell = \bigcup B_{i,a}$, such that $\ell = (2h+1) \cdot a + (i \mod (2h+1))$, with $0 \leq a < t$ and $1 \leq i \leq p$.
The vertices of $B_{i,a}$ are ordered according to the order of their $H$-coordinates on the track $T_a$. 
More precisely, for any two vertices $u,v\in B_{i,a}$, we place $u$ before $v$ if and only if \[ u.y\prec_a v.y. \] The buckets can be constructed without sorting. To this aim, for each vertex $y\in V(H)$, we first construct the list \[ C_y=\{v\in V(G):v.y=y\}. \] We then scan each track $T_a$ according to its order $\prec_a$. Whenever a vertex $y\in T_a$ is visited, every vertex $v\in C_y$, with $v.x=x_i$, is appended to the bucket $B_{i,a}$. 
Since the vertices of $T_a$ are processed according to $\prec_a$, the vertices in every bucket $B_{i,a}$ are automatically stored in the desired order. Finally, we process the vertices $x_1,\ldots,x_p$ in their order along $P$. 
For every nonempty bucket $B_{i,a}$, we append all its vertices, preserving their current order, to the end of the track $L_{(2h+1)a+(i\bmod (2h+1))}$. Consequently, within each copy $x_i[H]$, the order of the vertices agrees with the order inherited from~$\mathcal L$. Furthermore, if $i<j$ and $i\equiv j \bmod  (2h+1)$, then, on every track shared by $x_i[H]$ and $x_j[H]$, all vertices of $x_i[H]$ precede all vertices of $x_j[H]$. Equivalently, for any two distinct vertices $u$ and $v$ assigned to the same track of $\Psi$, we have $u\prec_\Psi v$ if and only if \[ \operatorname{ind}_P(u.x)<\operatorname{ind}_P(v.x), \] or \[ \operatorname{ind}_P(u.x)=\operatorname{ind}_P(v.x) \quad\text{and}\quad u.y\prec_{t_{\mathcal L}(u.y)}v.y. \] Hence, the order on each track of $\Psi$ is lexicographic with respect to the path index and the order inherited from $\mathcal L$, but it is constructed by scanning and concatenating ordered lists rather than by sorting the vertices.
We represent every track, list, and bucket as a linked list with pointers to its first and last elements. Under this representation, appending one element to a list and concatenating two lists take constant time. Moreover, only nonempty buckets are created. We also assume that $\cal L$ has no empty tracks, and hence $t < |V(H)|$, and that $|V(P)| + |V(H)| \in O(|V(G)|)$. Under these assumptions, it follows immediately that the space and time complexity of \cref{alg:path-track-layout} are linear in $h \cdot t + |V(G)|$. In particular, in \cref{alg:path-track-layout}, lines \ref{lst:line:1}-\ref{lst:line:2} require $O(h \cdot t)$ time, lines \ref{lst:line:3}-\ref{lst:line:4} require $O(n)$ time, and lines \ref{lst:line:5}-\ref{lst:line:6} require $O(h \cdot t + n)$ time.

In what follows, we outline the algorithm supporting \cref{thm:main} and show that, when applied to $G$, it runs in $O(n + h \cdot r \cdot t + f_t(H))$ time, assuming that there exists an algorithm ${\cal A_H}$ that computes a $t$-track layout of $H$ in $f_t(H)$ time.
In particular, the algorithm of \cref{thm:main} works as follows.

First, it applies the algorithm ${\cal A_H}$ to compute a $t$-track layout $\cal L$ of $H$.
Second, it provides $\cal L$ together with $D$, path $P=(x_1,\ldots,x_r)$, and the integer $r$ to \cref{alg:path-track-layout}  (recall that $P^{r-1}\simeq K_r$, $D \subseteq K_r \boxtimes H$) to obtain a track layout $\Psi_D$ of $D$ on at most $d \leq r \cdot t$ tracks in $O(r \cdot t + n)$ time.
Third, it provides $\Psi_D$ together with $G$ (recall that $G \subseteq P^h \boxtimes D$) 
and the integer $h$ to \cref{alg:path-track-layout} to obtain a track layout $\Psi$ of $G$ on $(2h+1) \cdot d$ tracks in $O(h \cdot d + n)$ time; specifically, in order to apply \cref{alg:path-track-layout}, for each vertex $v$ of $G$ we leave the field $v.x$ unchanged and assume that the field $v.y$ is given by the pair $\langle y_j,z_q\rangle$ combining the original $y$- and $z$-coordinates of $v$. Hence, the algorithm outputs a track layout of $G$ on at most $(2h+1)\cdot r \cdot t $ tracks in $O(n + h \cdot r \cdot t +  f_t(H))$ time. In particular, for planar graphs, we have that $h=1$, $r=3$, and $H$ belongs to the class of planar 3-trees, for which track layouts with at most $25$ tracks can be computed in linear time~\cite{DBLP:journals/jgaa/Pupyrev20}. Therefore, a track layout of any planar graph on at most $225$ tracks can be computed in linear time.

\section{Conclusions}
\label{sec:conclusions}

In this work, we investigated the problem of computing track layouts for families of graphs that admit a product structure. Our main result establishes new bounds on the track number of several broad graph families, including  genus-$k$, $k$-planar, $k$-framed, $k$-map, and $k$-string graphs. Moreover, in contrast to previous approaches, which are rather technically intricate (see, e.g.,~\cite{DBLP:journals/jgaa/Pupyrev20}), the framework developed in this work is both elegant and insightful, while also providing a solid foundation for further investigations. Most notably, it translates into a simple algorithm that needs no complex data structures and runs in linear time for planar graphs.

We conclude by highlighting and discussing two relevant open questions:

\begin{itemize}
\item A central open problem is to improve the upper bound of $(k+1)(2^{k+1}-2)^k$ on the track number of $k$-trees by Wiechert~\cite{DBLP:journals/combinatorics/Wiechert17}. Any progress in this direction would immediately translate into corresponding improved bounds for the track number of most graph families listed in \cref{tbl:results}. In particular, the case $k = 3$ under the additional assumption of planarity is of special interest, as several graph classes can be expressed in terms of the strong product of simple graphs and a planar $3$-tree. In this regard,
improving the current upper bound of $25$ on the track number of planar $3$-trees due to Pupyrev~\cite{DBLP:journals/jgaa/Pupyrev20} would have significant implications for these classes as well. 

\item As a second open problem, we highlight the challenge of improving the upper bound of $225$ on the track number of planar graphs, which our results match. Progress in this direction could provide valuable insights toward further improvements to the bounds established in this work. Despite our efforts towards this goal, however, an improvement remains elusive, highlighting both the difficulty of the problem and its fundamental importance in this area of research.
\end{itemize}

\bigskip\noindent\textbf{Acknowledgments. }The work of M. A. Bekos is supported by the HFRI grant 26320.
The work of P.~Hlin\v{e}n\'y is supported by the project 26-21334S of the Czech Science Foundation.

\printbibliography
\end{document}